\documentclass[12pt]{article}

\usepackage[margin=1in]{geometry}
\usepackage{amsmath, amssymb, amsthm}
\usepackage{xspace}
\usepackage{enumerate}
\usepackage{booktabs}
\usepackage{longtable}
\usepackage{tikz}
\usepackage{natbib}
\usepackage[colorlinks=true,linkcolor=blue!80!black,citecolor=blue!80!black,urlcolor=blue!80!black]{hyperref}

\usetikzlibrary{positioning, fit, backgrounds, calc}

\newtheorem{theorem}{Theorem}
\newtheorem{lemma}[theorem]{Lemma}
\newtheorem{proposition}[theorem]{Proposition}
\newtheorem{corollary}[theorem]{Corollary}

\theoremstyle{definition}

\newtheorem{example}[theorem]{Example}

\theoremstyle{remark}
\newtheorem{remark}[theorem]{Remark}

\newtheorem*{axiom-ce}{Component Efficiency (\textbf{CE})}
\newtheorem*{axiom-bc}{Balanced Contributions (\textbf{BC})}
\newtheorem*{axiom-f}{Fairness (\textbf{F})}
\newtheorem*{axiom-bcplus}{Pair-wise Balanced Contributions (\textbf{BC\textsuperscript{+}})}
\newtheorem*{axiom-s}{Symmetry (\textbf{S})}
\newtheorem*{axiom-l}{Linearity (\textbf{L})}

\newcommand{\bce}{\mathrm{BCE}}
\newcommand{\jw}{\mathrm{JW}}
\newcommand{\fce}{\mathrm{FCE}}
\newcommand{\Sh}{\mathrm{Sh}}
\newcommand{\My}{\mathrm{My}}

\newcommand{\CE}{component efficiency~(\textbf{CE})\xspace}
\newcommand{\BC}{balanced contributions~(\textbf{BC})\xspace}
\newcommand{\Fair}{fairness~(\textbf{F})\xspace}
\newcommand{\BCplus}{pair-wise balanced contributions~(\textbf{BC\textsuperscript{+}})\xspace}
\newcommand{\Sym}{symmetry~(\textbf{S})\xspace}

\tikzset{
	player/.style={circle, draw, thick, fill=white, inner sep=0pt, minimum size=20pt, font=\small},
	playerT/.style={circle, draw, thick, fill=blue!12, inner sep=0pt, minimum size=20pt, font=\small\bfseries},
	edge/.style={thick, gray!70},
}

\begin{document}

\title{Balanced Contributions in Networks and Games with Externalities}

\author{Frank Huettner\thanks{SKK GSB, Sungkyunkwan University, Seoul, South Korea. Email: mail@frankhuettner.de}}

\date{\today}

\maketitle

\begin{abstract}
For networks with externalities, where each component's worth may depend on the full network structure, balanced contributions and fairness lead to distinct component-efficient allocation rules.
We characterize the unique component-efficient allocation rule satisfying balanced contributions --- the BCE rule.
Existence is the main challenge: balanced contributions must hold on every edge, but the construction uses only spanning-tree edges.
A cycle-sum identity bridges this gap by reducing balanced contributions on non-tree edges to relations in proper subnetworks.
The BCE rule coincides with the Myerson value for TU games and with its generalization by Jackson--Wolinsky for network games without externalities, it recovers the externality-free value on the complete network, and --- unlike the fairness-based FCE rule --- it does not reduce to a graph-free formula applied to the graph-restricted game.


\end{abstract}

%% ============================================================
\section{Introduction}\label{sec:intro}
%% ============================================================

\citet{myerson1977b} introduced a graph-based model of cooperation in which players are linked by a communication network, and only connected groups can cooperate effectively.
For classical TU games, \citet{myerson1977b} characterized an allocation rule by \emph{fairness} and component efficiency; \citet{myerson1980} showed that component efficiency and \emph{balanced contributions} yields the same value.
Both axioms require connected players to benefit equally from their connection, but they impose different threats: fairness concerns removing a single link, while balanced contributions concerns a player's complete withdrawal.\footnote{\citet{osborne1994}, Ch.~14.4, derive balanced contributions from objections and counterobjections to proposed payoffs: player~$i$ objects by threatening to leave the game, and player~$j$ counters by noting that $j$'s own departure would hurt~$i$ at least as much. The axiom requires that every such objection has a counterobjection.}

\citet{jackson1996} enriched the model by allowing the worth of a component to depend on the full network structure, not just the coalition itself, and characterized the allocation rule satisfying component efficiency and fairness; \citet{slikker2007} showed that balanced contributions again yields the same rule.
\citet{navarro2007} extended the model further, permitting externalities across components, and characterized the allocation rule satisfying fairness and component efficiency (henceforth FCE rule).

With externalities, balanced contributions can capture indirect dependences that fairness misses, and leads to different payoffs.
For example, consider the networks $g$ and $g'$ in Figure~\ref{fig:intro-bc-vs-f}, and suppose player~$3$ gets a dollar whenever players~$1$ and~$2$ are linked; no other worth is generated.
In~$g$, player~$3$ is isolated and keeps the dollar by component efficiency.
In~$g'$, the link~$\{1,3\}$ connects player~$3$ to player~$1$, whose presence sustains the $\{1,2\}$-partnership: if player~$1$ withdraws, the partnership is destroyed, and player~$3$'s payoff vanishes.
Balanced contributions accounts for this dependence and redistributes payoffs; fairness, which only considers removing the single link~$\{1,3\}$ when balancing players~$1$ and~$3$, does not.

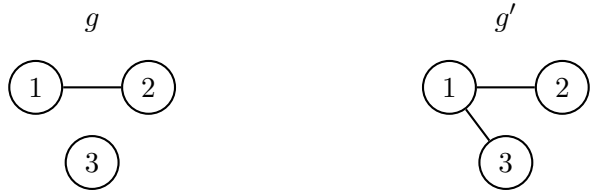
\begin{figure}[ht]
	\centering
	\begin{tikzpicture}[scale=1.0, every node/.style={font=\small}]
		\begin{scope}
			\node[anchor=south, font=\small\bfseries] at (0.75,1.6) {$g$};
			\node[player] (1) at (0,1) {$1$};
			\node[player] (2) at (1.5,1) {$2$};
			\node[player] (3) at (0.75,0) {$3$};
			\draw[thick] (1) -- (2);
		\end{scope}
		\begin{scope}[xshift=5.5cm]
			\node[anchor=south, font=\small\bfseries] at (0.75,1.6) {$g'$};
			\node[player] (1b) at (0,1) {$1$};
			\node[player] (2b) at (1.5,1) {$2$};
			\node[player] (3b) at (0.75,0) {$3$};
			\draw[thick] (1b) -- (2b);
			\draw[thick] (1b) -- (3b);
		\end{scope}
	\end{tikzpicture}
	\caption{Player~$3$ gets a dollar whenever players~$1$ and~$2$ are linked. In~$g$, player~$3$ is isolated and keeps the dollar. In~$g'$, the link~$\{1,3\}$ exposes player~$3$ to player~$1$'s threats. Balanced contributions forces equal split $(\tfrac{1}{3}, \tfrac{1}{3}, \tfrac{1}{3})$ in~$g'$; fairness leaves the payoff at $(0, 0, 1)$.}
	\label{fig:intro-bc-vs-f}
\end{figure}

However, an allocation rule satisfying balanced contributions and component efficiency in the presence of externalities has not been identified. The present paper fills this gap: we characterize the unique such rule (henceforth BCE rule), showing existence and uniqueness, and constructing it explicitly (Theorem~\ref{thm:main}). To this end, we impose balanced contributions only for directly connected players---a weaker requirement than pair-wise balanced contributions for all connected pairs, which we show to be incompatible with component efficiency. Moreover, it turns out that component efficiency, balanced contributions, and fairness are jointly incompatible in the presence of externalities.

Whereas uniqueness follows from an argument going back to \citet{myerson1977b} (within each component, one equation from component efficiency and $|C|-1$ equal gains equations along a spanning-tree pin down all $|C|$ payoffs), existence is the actual challenge.
Without externalities, existence is clear: the Shapley value of the graph-restricted game satisfies balanced contributions for all pairs of players, not just adjacent ones \citep{slikker2007}. With externalities, this pair-wise property fails, and the allocation rule defined via tree-edge balanced contributions equations must be shown to also satisfy balanced contributions on every non-tree edge.
This is established through the \emph{cycle-sum identity} (Lemma~\ref{lem:cycle-main}), which expresses balanced contributions on a non-tree edge as a signed sum of balanced contributions residuals in strict subnetworks, each zero by induction.

Section~\ref{sec:prelim} introduces the network game setting and the axioms.
Section~\ref{sec:main} constructs the BCE rule, states and proves the cycle-sum identity, and proves the main theorem.
Section~\ref{sec:consequences} investigates further properties and derives consequences. In particular, we recover the externality-free value \citep{phamdo2007,declippel2008} on the complete network without imposing a specific independence of externalities.
Section~\ref{sec:pff} connects to games in partition function form (PFF), showing how the BCE and FCE characterizations extend to this framework, and that the BCE value---unlike the FCE value---does not reduce to a graph-free formula applied to the graph-restricted PFF game.

%% ============================================================
\section{Preliminaries}\label{sec:prelim}
%% ============================================================

Let $N = \{1, \ldots, n\}$ be a finite player set.
A \emph{network} (undirected graph) $g$ on~$N$ is a set of unordered pairs $\{i,j\}$ with $i,j \in N$, $i \neq j$, called \emph{links}.
Write $g^N = \bigl\{\{i,j\} : i,j \in N,\, i \neq j\bigr\}$ for the \emph{complete} network, and $\mathcal{G}$ for the set of all networks on~$N$.

For $S \subseteq N$, write $g|_S = \bigl\{\{k,l\} \in g : k \in S,\, l \in S\bigr\}$ for the \emph{subnetwork induced by~$S$}, and $S/g$ for the partition of~$S$ into connected components of~$g|_S$.
In particular, $N/g$ is the partition of~$N$ into connected components of~$g$.
For a partition~$\mathcal{P}$ of~$N$, write $g|_{\mathcal{P}} = \bigcup_{B \in \mathcal{P}} g|_B$ for the subnetwork of~$g$ keeping only links between players in the same block of~$\mathcal{P}$.
For a subset $D \subseteq N$, write
\[
	g_{-D} \;=\; \bigl\{\{k,l\} \in g : k \notin D,\, l \notin D\bigr\}
\]
for the network with all links incident to any player in~$D$ deleted.
For a single player~$j$, write $g_{-j} = g_{-\{j\}}$.
Since link removal is set-theoretic,
$g_{-i,-j} \;=\; g_{-j,-i}.$

A \emph{cycle} of length $\ell \geq 3$ in~$g$ is a sequence of distinct players $i_0 \to i_1 \to \cdots \to i_{\ell-1} \to i_0$ with each consecutive pair linked in~$g$.
A network without cycles is a \emph{forest}; a connected forest is a \emph{tree}.

We adopt the framework of \citet{navarro2007}.
A \emph{worth function}\footnote{Following \citet{maschler2013}, we use \emph{worth} rather than \emph{value} to avoid confusion with allocation rules such as the Shapley value or the Myerson value.} is a map $w$ that assigns to each pair $(C, g)$,
where $C$ is a connected component of~$g$, a real number $w(C, g) \in \mathbb{R}$.
The worth $w(C, g)$ is available to the members of component~$C$ when the network is~$g$.
Crucially, $w(C, g)$ may depend on the full network~$g$, not only on the internal
structure $g|_C$; this captures \emph{externalities} across components.
We write $\mathcal{W}$ for the space of all worth functions.

An \emph{allocation rule} $\varphi\colon \mathcal{W}\times\mathcal{G}\to\mathbb{R}^N$
assigns a payoff $\varphi_i(w,g)\in\mathbb{R}$ to each worth function $w\in\mathcal{W}$,
network $g\in\mathcal{G}$, and player $i\in N$.
Note that both the game~$w$ and the allocation rule may depend on~$g$.

\subsection{Myerson's Equal Gains Axioms and Externalities}

Generalizing the axioms introduced by \citet{myerson1977a,myerson1980} to networks with externalities is straightforward and gives rise to the following axioms for allocation rules.

\begin{axiom-ce}\label{def:ce}
For every network~$g$ and every connected component $C \in N/g$:
\begin{equation}\label{eq:ce}
	\sum_{i \in C} \varphi_i(w, g) \;=\; w(C, g).
\end{equation}
\end{axiom-ce}

Component efficiency conveys the idea that disconnected players cannot generate surplus together, and that all jointly created worth is distributed amongst connected players.

\begin{axiom-bc}\label{def:bc}
For every network~$g$ and every link $\{i,j\} \in g$:
\begin{equation}\label{eq:bc}
	\varphi_i(w, g) - \varphi_i(w, g_{-j}) \;=\; \varphi_j(w, g) - \varphi_j(w, g_{-i}).
\end{equation}
\end{axiom-bc}

Balanced contributions requires that the gain player~$i$ obtains from player~$j$'s network presence equals the gain player~$j$ obtains from player~$i$'s presence.
The threatened action is \emph{complete withdrawal}: $g_{-j}$ removes \emph{all} links of player~$j$, not just the single link $\{i,j\}$. In contrast, the next axiom only considers the removal of one link.

\begin{axiom-f}\label{def:fairness}
For every network~$g$ and every link $\{i,j\} \in g$:
\[
	\varphi_i(w, g) - \varphi_i(w, g \setminus \{i,j\}) \;=\; \varphi_j(w, g) - \varphi_j(w, g \setminus \{i,j\}).
\]
\end{axiom-f}

Whereas \BC~and \Fair~apply to players who are directly connected by a link, the following property also applies to players who are indirectly connected.

\begin{axiom-bcplus}\label{def:bc-plus}
For every network~$g$ and players in the same component $\{i,j\} \in C \in N/g$:
\begin{equation}
	\varphi_i(w, g) - \varphi_i(w, g_{-j}) \;=\; \varphi_j(w, g) - \varphi_j(w, g_{-i}).
\end{equation}
\end{axiom-bcplus}

Clearly, \BCplus implies \BC. Yet, the difference appears negligible in the absence of externalities. Specifically, the Myerson value for TU games enriched by a network is characterized by \CE and \BC.\footnote{A \emph{TU game} $\bar{v}$ assigns a worth $\bar{v}(S) \in \mathbb{R}$ to each coalition~$S \subseteq N$, with $\bar{v}(\varnothing) = 0$. A TU game with a network is a pair $(\bar{v}, g)$. The \emph{Myerson value} \citep{myerson1977a} is $\My_i(\bar{v}, g) = \Sh_i(\bar{v}^g)$, where $\bar{v}^g(S) = \sum_{C \in S/g} \bar{v}(C)$ is the graph-restricted TU game and $\Sh_i(\bar{v}) = \sum_{S \subseteq N \setminus \{i\}} (|S|!\,(n-|S|-1)!)/n!\,[\bar{v}(S \cup \{i\}) - \bar{v}(S)]$ is the Shapley value. On TU games with networks, \CE reads $\sum_{i \in C} \varphi_i(\bar{v}, g) = \bar{v}(C)$ for every component~$C \in N/g$, and \BC reads $\varphi_i(\bar{v}, g) - \varphi_i(\bar{v}, g_{-j}) = \varphi_j(\bar{v}, g) - \varphi_j(\bar{v}, g_{-i})$ for every link $\{i,j\} \in g$.} Moreover, the Myerson value satisfies the stronger \BCplus, which could therefore be used to characterize the Myerson value together with \CE. \cite{slikker2007} expands this insight to network worth functions that exhibit no externalities. As will be shown in section~\ref{subsec:edge-vs-pair-wise}, the difference is crucial for networks with externalities.

Moreover, combining either \BC or \Fair together with \CE is characteristic of the same allocation rule in the absence of externalities (see \citet{myerson1977a,myerson1980} for TU games with networks; and \citet{jackson1996} and \citet{slikker2007} for network worth function without externalities).\footnote{A network worth function~$w$ has \emph{no externalities} if $w(C, g) = w(C, g|_C)$ for every component~$C$ and network~$g$: the worth of~$C$ depends only on the internal structure $g|_C$, not on links among players outside~$C$. Every TU game~$\bar v$ with a network induces such a worth function via $w(C, g) := \bar v(C)$; yet, worth functions are generally richer than TU games because worth functions may give different numbers depending on which links inside~$C$ are present.} If we incorporate externalities, the axioms lead to different allocation rules. Specifically, the allocation rule characterized in the next Theorem by \citet{navarro2007}, and the allocation rule of our main result the next section are different.

\begin{theorem}[Navarro, 2007]\label{thm:navarro}
	There exists a unique allocation rule for worth functions, henceforth FCE rule, satisfying \CE and \Fair.
\end{theorem}

The key contribution of this theorem is the existence part. Whereas uniqueness is a straightforward adoption of known induction argument, establishing the existence of a component efficient allocation rule that satisfies the fairness property is not as easy, even against the background that \citet{feldman1996} establishes a related result for the similar yet different framework of partition function form games with networks (see Section~\ref{sec:pff}).

%% ============================================================
\section{Main Result}\label{sec:main}
%% ============================================================

We prove that \CE and \BC characterize a unique allocation rule and construct it explicitly. Uniqueness will be fairly straightforward: within each component~$C$ of~$g$, \CE provides one equation and \BC on the $|C|-1$ edges of any spanning tree provides $|C|-1$ further equations, giving $|C|$ independent equations for $|C|$ unknowns. \emph{Existence} is the actual challenge.
A spanning tree selects only $|C|-1$ of the \BC equations---enough to pin down the allocation rule---but \BC demands that the remaining equations (one per non-tree link) hold as well.

\subsection{Construction of the BCE rule}

For notational simplicity, we define the BCE rule by means of specific trees. A consequence of the uniqueness result will be that the actual choice of the tree to construct the BCE rule does not matter.

For a connected component~$C$ of~$g$, the \emph{minimal-index BFS tree} is the breadth-first search tree rooted at $r = \min(C)$, with neighbors visited in increasing index order.
Explicitly: initialize $S = \{r\}$ and process players in breadth-first order from~$r$.
When player~$i$ is processed, for each neighbor~$j$ of~$i$ in~$C$ with $j \notin S$, taken in increasing order, set $p(j) = i$ (the \emph{parent} of~$j$) and add~$j$ to~$S$.
We write $p(j)$ for the parent of a non-root player~$j$ in this tree.

The \emph{BCE rule} is defined by induction on the number of links in~$g$. \emph{Base case} ($g = \emptyset$):\; $\bce_i(w, \emptyset) = w\bigl(\{i\},\, \emptyset\bigr)$.

\noindent\emph{Induction hypothesis}: Assume the rule has been defined for all graphs with fewer than~$m$ links, for some $m \geq 1$.

\noindent\emph{Induction step}: Let~$g$ have~$m$ links.
For each component~$C$ of~$g$ with $|C| = 1$, set $\bce_i(w,g) = w(\{i\},g)$.
For each component~$C$ of~$g$ with $|C| \geq 2$, fix the minimal-index BFS tree rooted at $r = \min(C)$, define $\gamma_r = 0$, and for each non-root $j \in C$ in BFS order let
\begin{equation}\label{eq:bc-tree}
	\gamma_j \;=\; \gamma_{p(j)} \;-\; \bigl[\bce_{p(j)}(w,\, g_{-j}) \;-\; \bce_j(w,\, g_{-p(j)})\bigr].
\end{equation}
Here $\bce_{p(j)}(w, g_{-j})$ and $\bce_j(w, g_{-p(j)})$ are well-defined by the induction hypothesis (each has less than $m$ links, since $g_{-j}$ removes at least the link $\{j, p(j)\}$).
Then for each $i \in C$:
\begin{equation}\label{eq:formula-full}
	\bce_i(w, g) \;=\; \frac{1}{|C|}\biggl[w(C,\, g) \;-\; \sum_{k \in C} \gamma_k\biggr] \;+\; \gamma_i.
\end{equation}

Note that the offset $\gamma_j - \gamma_{p(j)}$ encodes the asymmetry in how~$j$ and $p(j)$ contribute to each other's payoffs in the subnetworks $g_{-p(j)}$ and $g_{-j}$.
Formula~\eqref{eq:formula-full} then distributes the residual $w(C,g) - \sum_k \gamma_k$ equally among all members of~$C$ and adds the individual offsets~$\gamma_i$. This ensures that the following axioms hold.

\begin{lemma}\label{lem:ce-tree-bc}
	The $\bce$ rule satisfies \CE and (\ref{eq:bc}) on every edge of the minimal-index BFS tree.
\end{lemma}

\begin{proof}
	\emph{\textbf{CE}.}\;
	Summing~\eqref{eq:formula-full} over $i \in C$:
	\[
		\sum_{i \in C} \bce_i(w, g)
		= |C| \cdot \frac{1}{|C|}\Bigl[w(C, g) - \sum_{k \in C} \gamma_k\Bigr] + \sum_{k \in C} \gamma_k
		= w(C, g).
	\]

	\emph{(\ref{eq:bc}) on tree edges.}\;
	For any tree edge $\{j, p(j)\}$, equation~\eqref{eq:formula-full} gives
	\[
		\bce_j(w,g) - \bce_{p(j)}(w,g) = \gamma_j - \gamma_{p(j)},
	\]
	which by~\eqref{eq:bc-tree} equals $\bce_j(w,g_{-p(j)}) - \bce_{p(j)}(w,g_{-j})$.
\end{proof}

Note that the definition of the BCE rule directly refers to (minimal) indices, and that the choice of the spanning tree is not symmetric. Hence, it is not obvious if the BCE rule satisfies the following axiom.

\begin{axiom-s}\label{def:symmetry}
For every permutation $\pi$ of~$N$:\footnote{The permutation of a game is given by $(\pi w)(S, g') = w(\pi^{-1}S, \pi^{-1}g')$, where $\pi^{-1}g' = \{\{\pi^{-1}(k),\pi^{-1}(l)\} : \{k,l\} \in g'\}$ and $\pi g = \bigl\{\{\pi(i),\pi(j)\} : \{i,j\} \in g\bigr\}$.}
$\varphi_{\pi(i)}(\pi w,\, \pi g) \;=\; \varphi_i(w, g).$
\end{axiom-s}

\subsection{The BCE rule is characterized by \BC and \CE}

We are now ready to state our main result. It expands Lemma~\ref{lem:ce-tree-bc} by establishing (\ref{eq:bc}) to any link, including links not used for the definition of the BCE rule. Moreover, it establishes uniqueness of the BCE rule. Since the characteristic axioms \BC and \CE do not distinguish based on player names, it transpires that the BCE rule is symmetric.
\begin{theorem}\label{thm:main}
	The BCE rule is the unique allocation rule satisfying \CE and \BC.

	Moreover, the definition of the BCE rule in \eqref{eq:bc-tree}--\eqref{eq:formula-full} yields the same allocation rule for any spanning forests used in place of the minimal-index BFS forest and the BCE rule satisfies \Sym.
\end{theorem}

We prepare the proof of Theorem~\ref{thm:main} with a lemma that relates balanced contributions on non-tree edges to balanced contributions in smaller networks. This will later allow for an induction argument establishing the main result.

\subsection{Cycle-sum identity}

For any link $\{i,j\} \in g$ and any $D \subseteq N \setminus \{i,j\}$, the \emph{BC residual} at network~$g_{-D}$ is
\begin{equation}\label{eq:bc-residual}
	\begin{split}
		R^{\varphi}_{w}(i,j;\, g_{-D}) \;:=\; & \bigl[\varphi_i(w,\, g_{-D}) - \varphi_i(w,\, g_{-D \cup \{j\}})\bigr]      \\
		                                      & -\; \bigl[\varphi_j(w,\, g_{-D}) - \varphi_j(w,\, g_{-D \cup \{i\}})\bigr].
	\end{split}
\end{equation}
Note that $R^{\varphi}_{w}(i,j;\, g) = 0$ for all links~$\{i,j\} \in g$ is just \BC.

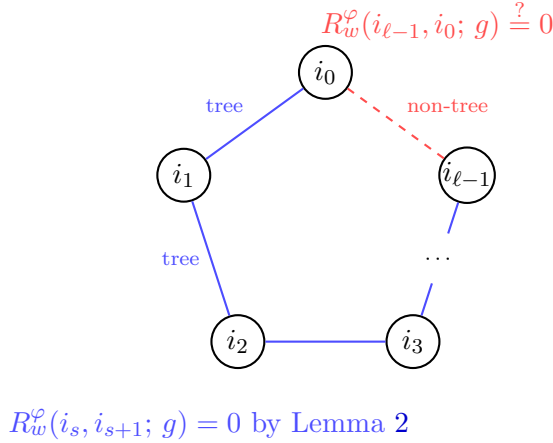
\begin{figure}[ht]
	\centering
	\begin{tikzpicture}[scale=1.1]
		\node[player] (0) at (90:1.8) {$i_0$};
		\node[player] (1) at (162:1.8) {$i_1$};
		\node[player] (2) at (234:1.8) {$i_2$};
		\node[player] (3) at (306:1.8) {$i_3$};
		\node[player] (4) at (18:1.8) {$i_{\ell-1}$};
		\draw[thick, blue!70] (0) -- (1) node[midway, above left, font=\scriptsize] {tree};
		\draw[thick, blue!70] (1) -- (2) node[midway, left, font=\scriptsize] {tree};
		\draw[thick, blue!70] (2) -- (3);
		\draw[thick, blue!70] (3) -- ($(3)!0.35!(4)$);
		\draw[thick, blue!70] (4) -- ($(4)!0.35!(3)$);
		\path (3) -- (4) node[midway, font=\scriptsize] {$\cdots$};
		\draw[thick, dashed, red!70] (4) -- (0) node[midway, above right, font=\scriptsize] {non-tree};
		\node[below=6pt, font=\small, text=blue!70] at (234:2.4) {$R^{\varphi}_{w}(i_s, i_{s+1};\, g) = 0$ by Lemma~\ref{lem:ce-tree-bc}};
		\node[above=6pt, font=\small, text=red!70] at (54:2.3) {$R^{\varphi}_{w}(i_{\ell-1}, i_0;\, g) \stackrel{?}{=} 0$};
	\end{tikzpicture}
	\caption{A fundamental cycle: tree edges (solid blue) satisfy balanced contributions (\ref{eq:bc}) by Lemma~\ref{lem:ce-tree-bc}; the cycle-sum identity reduces (\ref{eq:bc}) on the non-tree edge (dashed red) to BC-residuals in proper subnetworks.}
	\label{fig:cycle-sum}
\end{figure}

% Let $i_0 \to i_1 \to \cdots \to i_{\ell-1} \to i_0$ be a cycle of length $\ell \geq 3$ in a component of~$g$ (indices mod~$\ell$).
% Using the \BC residual~\eqref{eq:bc-residual}, write $R^{\varphi}_{w}(i_a, i_{a+1};\, g_{-D})$ for each cycle edge $\{i_a, i_{a+1}\}$ and $D \subseteq \{i_0, \ldots, i_{\ell-1}\} \setminus \{i_a, i_{a+1}\}$.

\begin{lemma}\label{lem:cycle-main}
	Fix network worth function $w$ and allocation rule $\varphi$.
	For any cycle $Z = i_0 \to i_1 \to \cdots \to i_{\ell-1} \to i_0$ of length $\ell \geq 3$ in~$g$ (indices mod~$\ell$):
	\begin{equation}\label{eq:cycle-identity}
		\sum_{s=0}^{\ell-1} R^{\varphi}_{w}(i_s, i_{s+1};\, g)
		= \sum_{\{i_a, i_{a+1}\} \in Z} \;\sum_{\substack{D \subseteq Z \setminus \{i_a, i_{a+1}\} \\ |D| \geq 1}} (-1)^{|D|+1}\, R^{\varphi}_{w}(i_a, i_{a+1};\, g_{-D}).
	\end{equation}
\end{lemma}

\medskip

Since $w$ and $\varphi$ are fixed throughout this section, we omit them in $R(i,j;\, g_{-D})$ when no confusion can arise. We first illustrate the identity for three players before giving the proof.

\begin{example}[Three-player cycle]\label{ex:triangle}
	Consider players $1, 2, 3$ forming the triangle $1 \to 2 \to 3 \to 1$ (so that $i_0 = 1$, $i_1 = 2$, $i_2 = 3$, $\ell = 3$).
	A spanning tree consists of edges $\{1,2\}$ and $\{2,3\}$; the non-tree edge is $\{3,1\}$.

	The Left-hand side has three terms ($s = 0,1,2$, indices mod~$3$):
	\begin{align*}
		LHS & = \bigl[\varphi_2(w, g_{-1}) - \varphi_3(w, g_{-1})\bigr]        \\
		    & \quad + \bigl[\varphi_3(w, g_{-2}) - \varphi_1(w, g_{-2})\bigr]  \\
		    & \quad + \bigl[\varphi_1(w, g_{-3}) - \varphi_2(w, g_{-3})\bigr].
	\end{align*}

	The right-hand side reduces to three terms since the only nonempty $D$s are singletons:
	\[
		RHS = R(1,2;\, g_{-3}) \;+\; R(2,3;\, g_{-1}) \;+\; R(3,1;\, g_{-2}).
	\]
	Expanding via~\eqref{eq:bc-residual} and collecting:
	\begin{align*}
		R(1,2;\, g_{-3}) & = \varphi_1(w, g_{-3}) - \varphi_1(w, g_{-\{2,3\}}) - \varphi_2(w, g_{-3}) + \varphi_2(w, g_{-\{1,3\}}), \\
		R(2,3;\, g_{-1}) & = \varphi_2(w, g_{-1}) - \varphi_2(w, g_{-\{1,3\}}) - \varphi_3(w, g_{-1}) + \varphi_3(w, g_{-\{1,2\}}), \\
		R(3,1;\, g_{-2}) & = \varphi_3(w, g_{-2}) - \varphi_3(w, g_{-\{1,2\}}) - \varphi_1(w, g_{-2}) + \varphi_1(w, g_{-\{2,3\}}).
	\end{align*}
	The double-isolation terms (those with two players removed) cancel pairwise, leaving the LHS.
\end{example}

\begin{proof}[Proof of Lemma~\ref{lem:cycle-main}]
	Both sides are linear combinations of formal symbols $\varphi_{i_c}(w, g_{-T})$ with $T \subseteq \{i_0, \ldots, i_{\ell-1}\} \setminus \{i_c\}$.
	We show the coefficients match by comparing contributions term by term.

	\medskip\noindent\textit{Left-hand side.}\;
	The LHS involves only singletons $|T| = 1$:
	coefficient $-1$ at $\varphi_{i_c}(w, g_{-i_{c+1}})$ and $+1$ at $\varphi_{i_c}(w, g_{-i_{c-1}})$ (indices mod~$\ell$).

	\medskip\noindent\textit{Right-hand side.}\;
	Write $\chi_c(T) := \varphi_{i_c}(w, g_{-T})$ for brevity. The RHS becomes
	\[
		\sum_{a=0}^{\ell-1}
		\sum_{\substack{D \subseteq Z\setminus\{i_a,i_{a+1}\} \\ |D|\geq 1}}
		(-1)^{|D|+1}
		\bigl[
			\chi_a(D)
			- \chi_a(D\cup\{i_{a+1}\})
			- \chi_{a+1}(D)
			+ \chi_{a+1}(D\cup\{i_a\})
			\bigr].
	\]

	Fix $c$ and collect all contributions involving $\chi_c(T)$ for a given $T \subseteq Z\setminus\{i_c\}$.
	The symbol $\chi_c(T)$ can appear from the two cycle edges adjacent to~$i_c$: edge $ \{i_c, i_{c+1}\}$ (where $i_c$ is the first endpoint) and edge $\{i_{c-1}, i_c\}$ (where $i_c$ is the second endpoint).

	\smallskip\emph{From $\{i_c, i_{c+1}\}$ (first-endpoint role):}\;
	The term $\chi_c(D)$ contributes $(-1)^{|D|+1}$ when $D = T$ and $i_{c+1} \notin T$.
	The term $-\chi_c(D \cup \{i_{c+1}\})$ contributes $-(-1)^{|D|+1} = (-1)^{|T|+1}$ when $D = T \setminus \{i_{c+1}\}$ and $i_{c+1} \in T$, $|T| \geq 2$.
	Both cases carry sign $(-1)^{|T|+1}$ and cover all $T$ with $|T| \geq 1$, except $T = \{i_{c+1}\}$:
	(a)~if $i_{c+1} \notin T$, then $D = T$ works, requiring $|T| \geq 1$;
	(b)~if $i_{c+1} \in T$, then $D = T \setminus \{i_{c+1}\}$ works, requiring $|D| \geq 1$, i.e., $|T| \geq 2$.
	The only $T$ not covered is $T = \{i_{c+1}\}$, which fails~(a) because $i_{c+1} \in T$ and fails~(b) because $|T| < 2$.
	So the contribution from $\{i_c, i_{c+1}\}$ is
	\[
		\sum_{\substack{T : |T| \geq 1 \\ T \neq \{i_{c+1}\}}} (-1)^{|T|+1} \chi_c(T)
		\;=\; -S_c - \chi_c(\{i_{c+1}\}),
	\]
	where $S_c = \sum_{|T| \geq 1} (-1)^{|T|} \chi_c(T)$.

	\smallskip\emph{From $\{i_{c-1}, i_c\}$ (second-endpoint role):}\;
	An identical argument with reversed sign gives
	$S_c + \chi_c(\{i_{c-1}\})$.

	In total, the $S_c$ terms cancel, leaving $\chi_c(\{i_{c-1}\}) - \chi_c(\{i_{c+1}\}) = \varphi_{i_c}(w, g_{-i_{c-1}}) - \varphi_{i_c}(w, g_{-i_{c+1}})$. Summing over $c$ and reindexing gives
	\[
		\mathrm{RHS}
		= \sum_{c=0}^{\ell-1}\bigl[\varphi_{i_c}(w,g_{-i_{c-1}})-\varphi_{i_c}(w,g_{-i_{c+1}})\bigr]
		= \mathrm{LHS}. \qedhere
	\]
\end{proof}

\subsection{Proof of Theorem~\ref{thm:main}}

\begin{proof}
	We prove existence and uniqueness simultaneously by induction on the number of links~$m$.
	\emph{Induction Basis} ($m = 0$):\;
	\textbf{CE} gives $\varphi_i(w, \emptyset) = w(\{i\}, \emptyset)$ uniquely; \textbf{BC} is vacuous.
	This coincides with $\bce(w, \emptyset)$.
	\emph{Induction hypothesis}:\;
	For all networks with at most~$m$ links, the unique value satisfying \textbf{CE} and \textbf{BC} is~$\bce$.
	\emph{Induction step}:\;
	Let~$g$ have $m+1$ links.

	\smallskip\noindent\textit{Existence.}\;
	By the induction hypothesis, $\bce(w, g')$ is well-defined for every network~$g'$ with at most~$m$ links, so the offsets~$\gamma_k$ in~\eqref{eq:bc-tree} are well-defined.
	Lemma~\ref{lem:ce-tree-bc} gives \textbf{CE} and \textbf{BC} on every BFS tree edge.
	It remains to verify \textbf{BC} on non-tree edges.

	Let $\{i_{\ell-1}, i_0\}$ be a non-tree edge in a component~$C$.
	Together with the BFS tree it forms a fundamental cycle $Z = i_0 \to i_1 \to \cdots \to i_{\ell-1} \to i_0$ of length $\ell \geq 3$, where all edges except $\{i_{\ell-1}, i_0\}$ are tree edges, for which  $R^{\bce}_{w}(i_0, i_1;\, g) = \ldots = R^{\bce}_{w}(i_{\ell-2}, i_{\ell-1};\, g) = 0$.
	We need to show that $R^{\bce}_{w}(i_{\ell-1}, i_0;\, g) = 0$.
	Since the tree-edge residuals vanish, we have $$R^{\bce}_{w}(i_{\ell-1}, i_0;\, g) = \sum_{s=0}^{\ell-1} R^{\bce}_{w}(i_s, i_{s+1};\, g),$$
	which equals the LHS of the cycle-sum identity in Lemma~\ref{lem:cycle-main}. Now consider the corresponding RHS of the cycle-sum identity
	$$\mathrm{RHS} = \sum_{\{i_a, i_{a+1}\} \in Z} \;\sum_{\substack{D \subseteq Z \setminus \{i_a, i_{a+1}\} \\ |D| \geq 1}} (-1)^{|D|+1}\, R^{\bce}_{w}(i_a, i_{a+1};\, g_{-D}),
	$$
	where each residual $R^{\bce}_{w}(i_a, i_{a+1};\, g_{-D})$ with $|D| \geq 1$ refers to network~$g_{-D}$ with at most~$m$ links. By the induction hypothesis, BCE satisfies \textbf{BC} at~$g_{-D}$, which gives $R^{\bce}_{w}(i_a, i_{a+1};\, g_{-D}) = 0$, hence $\mathrm{RHS} = 0$ and $R^{\bce}_{w}(i_{\ell-1}, i_0;\, g) = 0$.

	\smallskip\noindent\textit{Uniqueness.}\;
	Let $\varphi$ satisfy \textbf{CE} and \textbf{BC}, and let~$C$ be a component of~$g$ with $|C| \geq 2$.
	Fix any spanning tree of~$C$.
	For each tree edge $\{j, p(j)\}$, \textbf{BC} gives
	$\varphi_j(w,g) - \varphi_j(w,g_{-p(j)}) = \varphi_{p(j)}(w,g) - \varphi_{p(j)}(w,g_{-j})$.
	Since $g_{-j}$ and $g_{-p(j)}$ have at most~$m$ links, the induction hypothesis gives $\varphi = \bce$ on these subnetworks, hence
	\[
		\varphi_j(w,g) - \varphi_{p(j)}(w,g) \;=\; \bce_j(w,g) - \bce_{p(j)}(w,g)
	\]
	for each tree edge.
	Propagating along the tree, $\varphi_j(w,g) - \varphi_r(w,g) = \bce_j(w,g) - \bce_r(w,g)$ for all $j \in C$.
	Substituting into \textbf{CE} of $\varphi$ and BCE,
	\[
		\sum_{k \in C} \varphi_k(w, g) \;=\; w(C, g) \;=\; \sum_{k \in C} \bce_k(w, g),
	\]
	yields $\varphi_i(w,g) = \bce_i(w,g)$ for all $i \in C$.

	\smallskip\noindent\textit{Forest-independence.}\;
	If $\varphi$ is the value produced by any other spanning forest, then $\varphi$ also satisfies \textbf{CE} and \textbf{BC} (by the same Lemma~\ref{lem:ce-tree-bc} and cycle-sum argument applied to that forest).
	The uniqueness argument above then gives $\varphi = \bce$.

	\smallskip\noindent\textit{Symmetry.}\;
	For any permutation~$\pi$, the rule $\varphi_i(w, g) := \bce_{\pi(i)}(\pi w, \pi g)$ satisfies \textbf{CE} and \textbf{BC}.
	By uniqueness, $\varphi = \bce$.
\end{proof}

%% ============================================================
\section{Properties of the BCE rule}\label{sec:consequences}
%% ============================================================

While linearity of the BCE rule in the network worth function is immediate from its construction, other properties require closer examination. We investigate these next and contrast the BCE rule with the FCE rule.\footnote{Linearity: For every network~$g$, the map $w \mapsto \varphi(w, g)$ is linear: $\varphi(\alpha w + \beta w', g) = \alpha\,\varphi(w, g) + \beta\,\varphi(w', g)$ for all worth functions $w, w'$ and scalars $\alpha, \beta$.}

\subsection{BCE rule fails pair-wise balanced contributions}\label{subsec:edge-vs-pair-wise}

The BCE rule satisfies \BC on every edge but not \BCplus on all connected pairs as demonstrated by the following example.

\begin{example}\label{ex:pair-bc-fails}
	Let $N = \{1,2,3,4\}$ and define $w(C, h) = 1$ if $3 \in C$ and players~$1$ and~$2$ belong to the same component of~$h$, and $w(C,h) = 0$ otherwise.
	Consider $g = \bigl\{\{1,2\},\{1,4\},\{3,4\}\bigr\}$.
	The BCE rule gives $\bce(w, g) = (0,\, 0,\, 1,\, 0)$.
	\begin{figure}[ht]
		\centering
		\begin{tikzpicture}[scale=1.0, every node/.style={font=\small}]
			% --- g ---
			\begin{scope}
				\node[anchor=south, font=\small\bfseries] at (0.75,1.6) {$g$};
				\node[player] (1) at (0,1) {$1$};
				\node[player] (4) at (1.5,1) {$4$};
				\node[player] (2) at (0,0) {$2$};
				\node[player] (3) at (1.5,0) {$3$};
				\draw[thick] (1) -- (2);
				\draw[thick] (1) -- (4);
				\draw[thick] (4) -- (3);
				\node[left, font=\scriptsize] at (-0.4,1) {$0$};
				\node[right, font=\scriptsize] at (1.9,1) {$0$};
				\node[left, font=\scriptsize] at (-0.4,0) {$0$};
				\node[right, font=\scriptsize] at (1.9,0) {$1$};
			\end{scope}
			% --- g_{-3} ---
			\begin{scope}[xshift=4.5cm]
				\node[anchor=south, font=\small\bfseries] at (0.75,1.6) {$g_{-3}$};
				\node[player] (1b) at (0,1) {$1$};
				\node[player] (4b) at (1.5,1) {$4$};
				\node[player] (2b) at (0,0) {$2$};
				\node[player, fill=gray!20] (3b) at (1.5,0) {$3$};
				\draw[thick] (1b) -- (2b);
				\draw[thick] (1b) -- (4b);
				\node[left, font=\scriptsize] at (-0.4,1) {$0$};
				\node[right, font=\scriptsize] at (1.9,1) {$0$};
				\node[left, font=\scriptsize] at (-0.4,0) {$0$};
				\node[right, font=\scriptsize] at (1.9,0) {$1$};
			\end{scope}
			% --- g_{-1} ---
			\begin{scope}[xshift=9cm]
				\node[anchor=south, font=\small\bfseries] at (0.75,1.6) {$g_{-1}$};
				\node[player, fill=gray!20] (1c) at (0,1) {$1$};
				\node[player] (4c) at (1.5,1) {$4$};
				\node[player] (2c) at (0,0) {$2$};
				\node[player] (3c) at (1.5,0) {$3$};
				\draw[thick] (4c) -- (3c);
				\node[left, font=\scriptsize] at (-0.4,1) {$0$};
				\node[right, font=\scriptsize] at (1.9,1) {$0$};
				\node[left, font=\scriptsize] at (-0.4,0) {$0$};
				\node[right, font=\scriptsize] at (1.9,0) {$0$};
			\end{scope}
		\end{tikzpicture}
		\caption{Failure of \BCplus for the non-adjacent pair $\{1,3\}$. Numbers next to players are BCE payoffs. \textbf{Left:} $\bce(w,g) = (0,0,1,0)$. \textbf{Middle:} isolating player~$3$ preserves player~$1$'s payoff at $0$. \textbf{Right:} isolating player~$1$ breaks the partnership, so player~$3$'s payoff drops from $1$ to $0$. The asymmetry $0 - 0 \neq 1 - 0$ violates \BCplus.}
		\label{fig:pair-bc-fails}
	\end{figure}
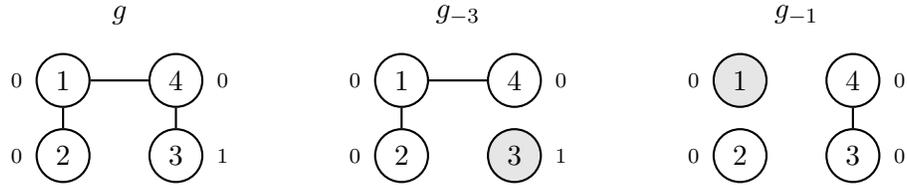
	\noindent
	Figure~\ref{fig:pair-bc-fails} shows the BCE payoffs at~$g$, $g_{-3}$, and $g_{-1}$.
	For the non-adjacent pair $\{1,3\}$:
	\[
		\bce_1(w, g) - \bce_1(w, g_{-3}) = 0 - 0
		\;\neq\;
		1 - 0 = \bce_3(w, g) - \bce_3(w, g_{-1}).
	\]
\end{example}
Together with Theorem~\ref{thm:main}, we get the following impossibility result.
\begin{corollary}\label{cor:pair-bc-incompatible}
	For networks with externalities, \CE is incompatible with \BCplus.
\end{corollary}

\subsection{BCE rule fails fairness}\label{sec:relation-fce}

The BCE rule does not coincide with the FCE rule as demonstrated by the next example.

\begin{example}\label{ex:bce-fails-fairness}
	Let $N = \{1,2,3\}$ and $w(C,h) = 1$ if $3 \in C$ and players~$1$ and~$2$ belong to the same component of~$h$; and $w(C,h) = 0$ otherwise.
	At $g = \{\{1,2\},\{1,3\}\}$, the BCE rule gives $\bce(w, g) = (\tfrac{1}{3},\, \tfrac{1}{3},\, \tfrac{1}{3})$, whereas the FCE rule assigns $\fce(w, g) = (0, 0, 1)$.
\end{example}

Recall that the BCE rule is characterized by \CE and \BC, whereas the FCE rule is characterized by \CE and \Fair. Since the two rules differ in general (Example~\ref{ex:bce-fails-fairness}), the BCE rule does not satisfy \Fair, and the FCE rule does not satisfy \BC. The three axioms are therefore mutually incompatible for networks with externalities.

\begin{corollary}\label{cor:bc-fairness-incompatible}
	For networks with externalities, \CE, \BC, and \Fair are mutually incompatible.
\end{corollary}

This is different in the absence of externalities, when the worth of a component depends only on the component's own internal structure---the \emph{no-externalities} setting of \citet{jackson1996}. The BCE rule and the FCE rule then coincide with the allocation rule of Jackson--Wolinsky, JW.

\begin{corollary}\label{cor:jw}
	Suppose $w(C, g) = w(C, g|_C)$ for every component~$C$ and network~$g$ (no externalities between components).
	Then $\bce_i(w, g) = \jw_i(w, g) = \fce_i(w, g)$ for all~$i$ and all~$g$.
	In particular, for a TU game~$\bar{v}$ with $w(C,g) = \bar{v}(C)$, the BCE rule coincides with the Myerson value $\My_i(\bar{v}, g) = \Sh_i(\bar{v}^g)$.
\end{corollary}

\begin{proof}
	JW satisfies \textbf{CE}, \textbf{F} \citep{jackson1996}, and \textbf{BC} \citep{slikker2007} for no-externality worth functions.\footnote{$\jw(w, g) = \Sh(\bar{v}^{JW})$, where the Shapley value $\Sh$ is applied to the TU game given by $\bar{v}^{JW}(S) = \sum_{C \in S/g} w(C, g)$.}
	The first claim follows with Theorems~\ref{thm:navarro} and~\ref{thm:main}.
	For the TU case, $\bar{v}^{JW}(S) = \sum_{C \in S/g} \bar{v}(C) = \bar{v}^g(S)$, so $\jw = \My$.
\end{proof}

\subsection{BCE rule on the complete network}\label{sec:complete}

On the complete network~$g^N$, every pair of players is adjacent, so the network imposes no restriction on cooperation, and \CE reduces to: $\sum_{i \in N} \bce_i(w, g^N) = w(N, g^N)$.

For $S \subseteq N$, let $g^S$ denote the network in which $S$ forms a complete subnetwork and each player outside~$S$ is isolated, so $(g^S)_{-j} = g^{S\setminus\{j\}}$ for $j \in S$.
Define the TU game $\bar{v}^{\mathrm{EF}}: 2^N \to \mathbb{R}$ by $\bar{v}^{\mathrm{EF}}(S) = w(S, g^S)$, and write $\bar{v}^{\mathrm{EF}}|_S$ for its restriction to subsets of~$S$.

\begin{proposition}\label{prop:bce-complete}
	For any worth function~$w$,\; $\bce_i(w, g^N) = \Sh_i(\bar{v}^{\mathrm{EF}})$.
\end{proposition}

\begin{proof}
	We show $\bce_i(w, g^S) = \Sh_i(\bar{v}^{\mathrm{EF}}|_S)$ by induction on~$|S|$.

	\emph{Induction basis case:} For $|S|=1$, every player is a singleton component, so \textbf{CE} gives $\bce_i(w, g^{\{i\}}) = w(\{i\}, \emptyset) = \bar{v}^{\mathrm{EF}}(\{i\}) = \Sh_i(\bar{v}^{\mathrm{EF}}|_{\{i\}})$.
	\emph{Induction hypothesis:}
	Assume $\bce_i(w, g^{S'}) = \Sh_i(\bar{v}^{\mathrm{EF}}|_{S'})$ for all $S' \subsetneq S$.

	\emph{Inductive step:}
	By \textbf{BC}, $\bce_i(w, g^S) - \bce_i(w, g^{S \setminus \{j\}}) = \bce_j(w, g^S) - \bce_j(w, g^{S \setminus \{i\}})$ for all $i,j \in S$.
	By the induction hypothesis, $\bce_i(w, g^{S \setminus \{j\}}) = \Sh_i(\bar{v}^{\mathrm{EF}}|_{S \setminus \{j\}})$, and using balanced contributions of the Shapley value gives $\bce_i(w, g^S) - \bce_j(w, g^S) = \Sh_i(\bar{v}^{\mathrm{EF}}|_S) - \Sh_j(\bar{v}^{\mathrm{EF}}|_S)$.
	Inserting into \textbf{CE} of BCE and efficiency of the Shapley value, $\sum_{i \in S} \bce_i(w, g^S) = \bar{v}^{\mathrm{EF}}(S) = \sum_{i \in S} \Sh_i(\bar{v}^{\mathrm{EF}}|_S)$ yields $\bce_i(w, g^S) = \Sh_i(\bar{v}^{\mathrm{EF}}|_S)$.
\end{proof}

The TU game $\bar{v}^{\mathrm{EF}}(S) = w(S, g^S)$ captures the natural ``insiders together, outsiders alone'' idea: each coalition~$S$ is evaluated in the network~$g^S$ where all outsiders are isolated singletons.
Proposition~\ref{prop:bce-complete} shows this value arises as a \emph{consequence} of \CE and \BC on~$g^N$ rather than an assumption.
In the partition function form (PFF) framework (see Section~\ref{sec:pff}), \citet{declippel2008} characterize~$\Sh(\bar{v}^{\mathrm{EF}})$ using a marginality axiom that somewhat isolates a player's payoff from her external effects: only what coalitions lose when she leaves towards a stand alone matters.\footnote{\citet{declippel2008} work with PFF games~$v(S, \mathcal{P})$. They define a player's \emph{intrinsic marginal contribution} to coalition~$S$ at partition~$\mathcal{P}$ as $v(S, \mathcal{P}) - v(S \setminus \{i\},\, \{S \setminus \{i\}, \{i\}\} \cup \mathcal{P}\setminus \{S\})$: player~$i$ leaves~$S$ to become an isolated singleton. \emph{Marginality} requires that a player's payoff depends on~$v$ only through these intrinsic marginal contributions.}
In contrast, the reduction to an externality-free allocation rule in Proposition~\ref{prop:bce-complete} emerges from \CE and \BC.

Note that Proposition~\ref{prop:bce-complete} is specific to~$g^N$; for a general network~$g$, $\bce_i(w, g)$ need not equal $\Sh_i(\bar{v}^{\mathrm{EF}})$.

\subsection{Invariance under graph-projection}\label{sec:graph-projection}

A network worth function~$w$ may encode information about links outside of some network~$g$ at hand --- for instance, $w(C,h)$ may depend on whether a particular link $\{i,j\} \notin g$ is present in~$h$. The BCE rule, however, only ever evaluates~$w$ at subnetworks of~$g$ to compute $\mathrm{BCE}(w, g)$, so that any information about links not in~$g$ is irrelevant. We make this precise by projecting~$w$ onto~$g$ and showing that the BCE rule is invariant under this projection.

Given a network worth function~$w$ and a network~$g$, define the \emph{$g$-projected worth function} $w^{g}$ by
\[
	w^{g}(C, h) \;=\; \sum_{S \in C/g} w(S,\, g|_{N/h}),
\]
where $C$ is a component of~$h$.
Here $g|_{N/h}$ takes the reference network~$g$ and strips out any edges that cross components of the current network~$h$.
Then, $w^{g}$ breaks each component~$C$ of~$h$ into its $g$-connected pieces and sums their worths at $g|_{N/h}$.

The projected worth function $w^g$ bakes in $g$'s structure upfront, discarding information the BCE rule would never use.
The following is a straightforward consequence of the inductive construction; its main role is to prepare the connection to PFF games in Section~\ref{sec:pff}.

\begin{corollary}\label{cor:graph-projection}
	For any network worth function~$w$ and network~$g$,
	\[
		\bce(w, g) \;=\; \bce(w^g, g).
	\]
\end{corollary}

\begin{proof}
	Every network visited by the recursive construction of $\bce(\cdot, g)$ is of the form~$g_{-D}$ for some $D \subseteq N$.
	Since $g_{-D}$ only removes edges of~$g$, every component~$C$ of~$g_{-D}$ is $g$-connected ($C/g = \{C\}$) and $g_{-D} = g|_{N/(g_{-D})}$.
	Therefore $w^g(C, g_{-D}) = w(C, g_{-D})$ at every level of the recursion, giving $\bce(w^g, g) = \bce(w, g)$.
\end{proof}

The FCE~rule is also invariant under graph-projection,  $\fce(w, g) = \fce(w^g, g)$, see Corollary~\ref{cor:fce-graph-projection} in Section~\ref{sec:pff-bce-vs-fce}.

\begin{example}\label{ex:graph-projection}
	Let $N = \{1,2,3\}$ and define $w(C, h) = 1$ if $3 \in C$ and $\{1,2\} \in h$, and $w(C,h) = 0$ otherwise: the worth function rewards components containing player~$3$ whenever the link~$\{1,2\}$ is present.
	Consider the two networks in Figure~\ref{fig:graph-projection-example}.
	\begin{figure}[!ht]
		\centering
		\begin{tikzpicture}[scale=1.0, every node/.style={font=\small}]
			\begin{scope}
				\node[anchor=south, font=\small\bfseries] at (1.1,0.6) {$g = \{\{1,2\}\}$};
				\node[player] (1) at (0.75,0) {$1$};
				\node[player] (2) at (2.25,0) {$2$};
				\node[player] (3) at (1.5,-1) {$3$};
				\draw[thick] (1) -- (2);
			\end{scope}
			\begin{scope}[xshift=5.5cm]
				\node[anchor=south, font=\small\bfseries] at (1.1,0.6) {$h = \{\{1,3\},\{2,3\}\}$};
				\node[player] (1b) at (0.75,0) {$1$};
				\node[player] (2b) at (2.25,0) {$2$};
				\node[player] (3b) at (1.5,-1) {$3$};
				\draw[thick] (1b) -- (3b);
				\draw[thick] (2b) -- (3b);
			\end{scope}
		\end{tikzpicture}
		\caption{Graph-projection can change~$w$ at networks~$h$ that are not subnetworks of~$g$: here $w(\{1,2,3\},h)=0$ but $w^g(\{1,2,3\},h)=1$. Yet the $\bce(w, g)$ only evaluates~$w$ on subnetworks of~$g$, where $w^g=w$, hence $\bce(w^g, g) = \bce(w, g)$.}
		\label{fig:graph-projection-example}
	\end{figure}
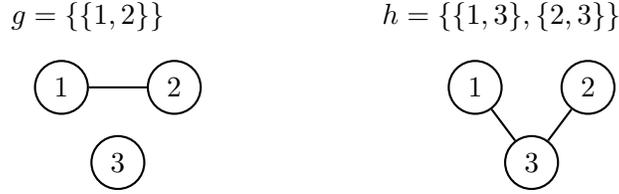
	Graph projection changes~$w$ at~$h$: the link~$\{1,2\}$ is absent in~$h$, so $w(\{1,2,3\}, h) = 0$, but the $g$-projection reads off~$g$ instead and returns $w^g(\{1,2,3\}, h) = w(\{3\}, g) + w(\{1,2\}, g) = 1$.
	Yet the proposition guarantees $\bce(w, g) = \bce(w^g, g)$, since the inductive construction only visits subnetworks of~$g$ where $w^g = w$.
\end{example}

%% ============================================================
\section{Connection to partition function form games}\label{sec:pff}
%% ============================================================

Sections~\ref{sec:prelim}--\ref{sec:consequences} work with network worth functions, i.e., $w(C,g)$ may depend on the full link structure of~$g$.
An alternative framework pairs a partition function form game---defined independently of any network---with a network~$g$.
This section clarifies how the characterizations of the BCE rule and the FCE rule carry over to this framework, and how the two values rely on the information of the network.

\subsection{BCE and FCE values for PFF games}\label{sec:pff-two-frameworks}

A \emph{partition function form (PFF) game} \citep{thrall1963} assigns a worth $v(C,\mathcal{P}) \in \mathbb{R}$ to every \emph{embedded coalition} $(C,\mathcal{P})$, where $\mathcal{P}$ is a partition of~$N$ and $C \in \mathcal{P}$.%
\footnote{A \emph{partition} of a finite set $N$ is a collection of nonempty, pairwise disjoint subsets of $N$ whose union is $N$.}
A PFF game~$v$ induces the worth function
$$w_{v}(C, g) = v(C,\, N/g).$$
The map $v \mapsto w_v$ is injective (every partition of~$N$ is realizable as $N/g$).
For any PFF-induced worth function~$w$, we write $v_w$ for the unique PFF game with $w_{v_w} = w$, so that $v_w$ and $w_v$ are inverses.

A \emph{PFF game with a network} \citep{feldman1996,navarro2007} is a pair $(v, g)$.
A \emph{value} $\Psi$ on PFF games with a network
assigns a payoff $\Psi_i(v,g)\in\mathbb{R}$ to each PFF game $v$, network $g$ and player $i$.
The \emph{FCE value} $\Psi^{\fce}(v, g) := \fce(w_v, g)$ is the unique value on PFF games with a network satisfying \CE and \Fair  \citep{feldman1996,navarro2007}.\footnote{\textbf{CE} on PFF games with networks: $\sum_{i \in C} \Psi_i(v, g) = v(C, N/g)$ for every component $C \in N/g$. \textbf{F} on PFF games with networks: $\Psi_i(v, g) - \Psi_i(v, g \setminus \{i,j\}) = \Psi_j(v, g) - \Psi_j(v, g \setminus \{i,j\})$ for every link $\{i,j\} \in g$.}

Analogously, we define the \emph{BCE value} by $$\Psi^{\bce}(v, g) := \bce(w_v, g).$$
This value inherits its characteristic properties from the BCE rule.

\begin{corollary}\label{cor:pff-characterization}
	$\Psi^{\bce}$ is the unique value on PFF games with a network satisfying \CE and \BC.\footnote{\textbf{BC} on PFF games with networks: $\Psi_i(v, g) - \Psi_i(v, g_{-j}) = \Psi_j(v, g) - \Psi_j(v, g_{-i})$ for every link $\{i,j\} \in g$.}
\end{corollary}

\begin{proof}
	\emph{Existence.}
	By Theorem~\ref{thm:main}, the BCE rule satisfies \textbf{CE} and \textbf{BC} for every worth function, in particular for~$w_v$.
	Since $\Psi^{\bce}_i(v,g) = \bce_i(w_v, g)$ by definition and $w_v(C,g) = v(C, N/g)$, \textbf{CE} of $\bce$ at $(w_v, g)$ yields $\sum_{i \in C} \Psi^{\bce}_i(v,g) = v(C, N/g)$ for every $C \in N/g$.
	Likewise, \textbf{BC} of $\bce$ at $(w_v, g)$ translates directly to \textbf{BC} for~$\Psi^{\bce}$, because $\Psi^{\bce}_i(v, g') = \bce_i(w_v, g')$ holds for every network~$g'$.
	\emph{Uniqueness.}
	Fix~$v$; the standard induction argument based on the number of links in~$g$ (as in the proof of Theorem~\ref{thm:main}) gives $\Psi(v, g) = \bce(w_v, g)$.
\end{proof}

\subsection{Graph restriction}\label{sec:pff-framework}

The \emph{graph-restricted PFF game} $v^g$ \citep{feldman1996,navarro2007} is defined by
\[
	v^g(C, \mathcal{P}) \;=\; \sum_{S \in C/g} v(S,\, g/\mathcal{P}),
	\qquad g/\mathcal{P} \;=\; \bigcup_{B \in \mathcal{P}} B/g.
\]
\begin{lemma}\label{lem:projection-restriction}
	$(w_v)^g = w_{v^g}$ for every PFF game~$v$ and network~$g$.
\end{lemma}

\begin{proof}
	For any component~$C$ of~$h$,
	\[
		(w_v)^g(C, h)
		\;=\; \sum_{S \in C/g} w_v(S,\, g|_{N/h})
		\;=\; \sum_{S \in C/g} v(S,\, N/(g|_{N/h})).
	\]
	Since $g|_{N/h}$ restricts~$g$ to edges within each block of~$N/h$, its connected components are exactly the $g$-components within each block, i.e., $N/(g|_{N/h}) = \bigcup_{B \in N/h} B/g = g/(N/h)$.
	Hence $(w_v)^g(C, h) = \sum_{S \in C/g} v(S,\, g/(N/h)) = v^g(C, N/h) = w_{v^g}(C, h)$.
\end{proof}

\begin{corollary}\label{cor:graph-restriction}
	$\Psi^{\bce}(v, g) = \Psi^{\bce}(v^g, g)$ for every PFF game~$v$ and network~$g$.
\end{corollary}

\begin{proof}
	$\Psi^{\bce}(v, g) = \bce(w_v, g) = \bce((w_v)^g, g)$ by Corollary~\ref{cor:graph-projection}, which equals $\bce(w_{v^g}, g) = \Psi^{\bce}(v^g, g)$ by Lemma~\ref{lem:projection-restriction}.
\end{proof}

A PFF game $v(C, \mathcal{P})$ can encode dependence on arbitrary partitions $\mathcal{P}$, including those not realizable as the component structure $N/h$ of any network. Graph restriction strips this away: $v^g(C, \mathcal{P})$ depends on $\mathcal{P}$ only through $g/\mathcal{P}$, so $v^g$ carries no more information than the worth function $w_{v^g} = (w_v)^g$. Graph-restriction invariance then implies that only $v^g$ matters to a value, not the original~$v$.
Conversely, every graph-projected worth function is PFF-induced, i.e., an allocation rule that does not distinguish between $w$ and its graph projection $w^g$ can equivalently be viewed as a value on PFF games with networks.

\begin{proposition}\label{rem:wg-pff}
	For any network worth function~$w$ and network~$g$, the projected worth function~$w^g$ is PFF-induced, i.e., there is a unique PFF game $v$ such that~$w^g = w_v$.
\end{proposition}

\begin{proof}
	The partition $C/g$ depends only on~$C$ and~$g$, not on~$h$.
	The network $g|_{N/h} = \bigcup_{B \in N/h} g|_B$ is determined by~$N/h$.
	Hence, $w^g(C,h) = \sum_{S \in C/g} w(S, g|_{N/h})$ depends on~$h$ only through the partition~$N/h$.
	The inducing PFF game is $v(C, \mathcal{P}) = \sum_{S \in C/g} w(S, g|_{\mathcal{P}})$.
	Uniqueness holds because $v \mapsto w_v$ is injective: every partition of~$N$ is realizable as~$N/h$ for some network~$h$.
\end{proof}

\begin{remark}
	Graph-projection invariance (for rules) and graph-restriction invariance (for values) are \emph{consequences} of \CE and \BC (Corollary~\ref{cor:graph-projection} and Corollary~\ref{cor:pff-characterization}); for the FCE side, the analogous result follows from \citeauthor{navarro2007}'s \citeyearpar{navarro2007} formula~\eqref{eq:navarro-fce}. They are not assumptions.
	If one were to restrict attention to rules or values that are a priori invariant under graph projection/restriction, then studying worth functions or PFF games would be interchangeable: in one direction, $w^g$ is PFF-induced (Proposition~\ref{rem:wg-pff}); in the other, every PFF game~$v$ embeds into worth functions via the injective map $v \mapsto w_v$ (every partition of~$N$ is realizable as $N/g$).
	However, Theorem~\ref{thm:main} and Corollary~\ref{cor:pff-characterization}, as well as Theorem~\ref{thm:navarro} and \citeauthor{feldman1996}'s \citeyearpar{feldman1996} characterization, impose no such restriction: they characterize their respective allocation rules from the axioms alone, making each a genuine result.
\end{remark}

\subsection{The FCE value reduces to a graph-free value; the BCE value does not}\label{sec:pff-bce-vs-fce}

A partition $\mathcal{Q}$ is \emph{finer} than $\mathcal{P}$, written $\mathcal{Q} \preccurlyeq \mathcal{P}$, if every block of~$\mathcal{Q}$ is a subset of some block of~$\mathcal{P}$.
\citet{myerson1977a} introduces the \emph{$\preccurlyeq$-unanimity games} $u^{\preccurlyeq}_{(T,\mathcal{Q})}$, defined by $u^{\preccurlyeq}_{(T,\mathcal{Q})}(C,\mathcal{P}) = 1$ if $T \subseteq C$ and $\mathcal{Q} \preccurlyeq \mathcal{P}$, and $0$ otherwise.
Every PFF game has a unique representation $v = \sum_{(T,\mathcal{Q})} b^{\preccurlyeq}_{(T,\mathcal{Q})}(v)\,u^{\preccurlyeq}_{(T,\mathcal{Q})}$.
The \emph{PFF value} $\Phi^{\preccurlyeq}$ \citep{myerson1977a} distributes each coefficient equally among the members of~$T$:
$$\Phi^{\preccurlyeq}_i(v) = \sum_{(T,\mathcal{Q}):\, i \in T} \frac{b^{\preccurlyeq}_{(T,\mathcal{Q})}(v)}{|T|}, \quad \textrm{where } b^{\preccurlyeq}_{(T,\mathcal{Q})}(v) = v(T,\mathcal{Q}) - \sum_{\substack{(C,\mathcal{P})\preccurlyeq (T,\mathcal{Q}) \\ (C,\mathcal{P}) \neq (T,\mathcal{Q})}}b^{\preccurlyeq}_{(C,\mathcal{P})}(v).$$

\citet{feldman1996} shows $\Psi^{\fce}(v, g) = \Phi^{\preccurlyeq}(v^g)$: the FCE value applies the PFF value $\Phi^{\preccurlyeq}$ to the graph-restricted game.
The network enters only through~$v^g$; after graph restriction, it plays no further role.
\citet{navarro2007} (Theorem~4.2) extends this to general worth functions:
\begin{equation}\label{eq:navarro-fce}
	\fce(w, g) \;=\; \Phi^{\preccurlyeq}(v_{w^g}),
	\qquad \textrm{where }v_{w^g}(C, \mathcal{P}) \;=\; \sum_{S \in C/g} w(S,\, g|_{\mathcal{P}}).
\end{equation}

As an immediate consequence, the FCE rule is invariant under graph-projection:

\begin{corollary}\label{cor:fce-graph-projection}
	For any network worth function~$w$ and network~$g$,\; $$\fce(w^g, g) = \fce(w, g).$$
\end{corollary}

\begin{proof}
	The projection is idempotent: $(w^g)^g = w^g$, so $v_{(w^g)^g} = v_{w^g}$.
	Applying~\eqref{eq:navarro-fce} gives
	$\fce(w^g, g) = \Phi^{\preccurlyeq}(v_{(w^g)^g}) = \Phi^{\preccurlyeq}(v_{w^g}) = \fce(w, g)$.
\end{proof}

The BCE value genuinely requires the network.
The \emph{externality-free value}~$\Phi^{\mathrm{EF}}$ \citep{phamdo2007,declippel2008} is given by $$\Phi^{\mathrm{EF}}_i(v) = \Sh_i(\bar{v}^{\mathrm{EF}}), \quad \textrm{where }\bar{v}^{\mathrm{EF}}(S) = v(S,\, \{S\} \cup \{\{j\} : j \notin S\}).$$
On the complete network, $v^{g^N} = v$ and Proposition~\ref{prop:bce-complete} gives:

\begin{corollary}[BCE on $g^N$]\label{cor:bce-complete-pff}
	For any PFF game~$v$,\; $\Psi^{\bce}_i(v, g^N) = \Phi^{\mathrm{EF}}_i(v)$.
\end{corollary}

On other networks, no such reduction exists: the BCE value cannot be written as a function of~$v^g$ alone.

\begin{example}\label{ex:no-phi}
	Let $N = \{1,2,3\}$ and $v = u^{\preccurlyeq}_{(\{3\},\,\{\{1,2\},\{3\}\})}$, i.e., $v(C,\mathcal{P}) = 1$ if $\{3\} \subseteq C$ and players~$1$ and~$2$ belong to the same block of~$\mathcal{P}$, and $v(C,\mathcal{P}) = 0$ otherwise.\footnote{This is $v_{w^g}$ for the worth function~$w$ of Example~\ref{ex:graph-projection}, which checks for the \emph{edge} $\{1,2\}$ rather than for $1$ and $2$ being in the same partition block. Moreover: $\bce(w, g) = (0,0,1)$, $\bce(w, g')  = (\tfrac{1}{3},\tfrac{1}{3},\tfrac{1}{3})$, and $\fce(w, g) = \fce(w, g') = (0,0,1)$.}
	\begin{figure}[ht]
		\centering
		\begin{tikzpicture}[scale=1.0, every node/.style={font=\small}]
			\begin{scope}
				\node[anchor=south, font=\small\bfseries] at (0.75,1.6) {$g$};
				\node[player] (1) at (0,1) {$1$};
				\node[player] (2) at (1.5,1) {$2$};
				\node[player] (3) at (0.75,0) {$3$};
				\draw[thick] (1) -- (2);
				\node[below, gray, font=\scriptsize] at (0.75,-0.4) {$\Psi^{\bce}(v, g) = (0,\, 0,\, 1)$};
				\node[below, gray, font=\scriptsize] at (0.75,-0.85) {$\Psi^{\fce}(v, g) = (0,\, 0,\, 1)$};
			\end{scope}
			\begin{scope}[xshift=5.5cm]
				\node[anchor=south, font=\small\bfseries] at (0.75,1.6) {$g'$};
				\node[player] (1b) at (0,1) {$1$};
				\node[player] (2b) at (1.5,1) {$2$};
				\node[player] (3b) at (0.75,0) {$3$};
				\draw[thick] (1b) -- (2b);
				\draw[thick] (1b) -- (3b);
				\node[below, gray, font=\scriptsize] at (0.75,-0.4) {$\Psi^{\bce}(v, g') = (\tfrac{1}{3},\, \tfrac{1}{3},\, \tfrac{1}{3})$};
				\node[below, gray, font=\scriptsize] at (0.75,-0.85) {$\Psi^{\fce}(v, g') = (0,\, 0,\, 1)$};
			\end{scope}
		\end{tikzpicture}
		\caption{Two networks with $v^g = v^{g'}$ for $v = u^{\preccurlyeq}_{(\{3\},\,\{\{1,2\},\{3\}\})}$. The FCE value is the same for both; the BCE value differs. To the FCE value, it does not matter whether player~$1$ can communicate the threat of deleting $\{1,2\}$ to player~$3$; the BCE value accounts for this through balanced contributions.}
		\label{fig:no-phi}
	\end{figure}
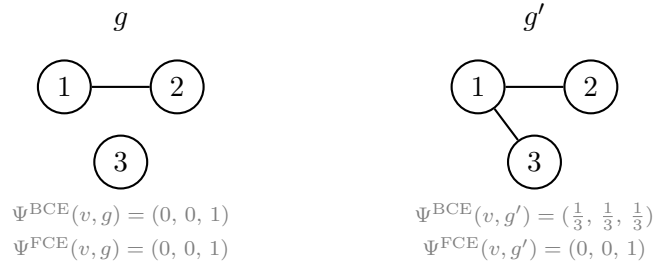
	Consider the two networks in Figure~\ref{fig:no-phi}: $g = \bigl\{\{1,2\}\bigr\}$ and $g' = \bigl\{\{1,2\}, \{1,3\}\bigr\}$.
	Since both networks contain the link~$\{1,2\}$, $v^g = v^{g'}$.
	Yet the BCE value distinguishes them: at~$g$, player~$3$ is a singleton component with worth~$1$, so $\Psi^{\bce}(v, g) = (0,0,1)$.
	At~$g'$, all three form one component; balanced contributions forces the equal split $\Psi^{\bce}(v, g') = (\tfrac{1}{3},\tfrac{1}{3},\tfrac{1}{3})$.
\end{example}

\begin{corollary}\label{cor:no-phi}
	There exist a PFF game~$v$ and networks $g, g'$ with $v^g = v^{g'}$ but $\Psi^{\bce}(v, g) \neq \Psi^{\bce}(v, g')$.
\end{corollary}

\begin{remark}
	The FCE value, by contrast, factors through the graph-restricted game: $\Psi^{\fce}(v, g) = \Phi^{\preccurlyeq}(v^g)$.
	In Example~\ref{ex:no-phi}, $v^g = v^{g'}$ implies $\Psi^{\fce}(v, g) = \Psi^{\fce}(v, g')$.
	The link~$\{1,3\}$ changes the balanced-contributions equations --- information that graph restriction discards and~$\Phi^{\preccurlyeq}$ never sees.
\end{remark}

%% ============================================================
\section{Concluding remarks}\label{sec:conclusion}
%% ============================================================

The BCE rule is the unique allocation rule satisfying \CE and \BC for network games with externalities.
It coincides with the Jackson--Wolinsky's generalization of the Myerson value \citep{jackson1996,myerson1977b} in the absence of externalities, and it recovers the externality-free value \citep{phamdo2007,declippel2008} on the complete network.
Unlike the FCE rule, which factors through the graph-restricted game and a known partition function form value, the BCE rule is defined purely through an inductive construction and admits no known closed-form expression, dividend decomposition, or potential function.

This limitation appears structural rather than incidental.
The Hart--Mas-Colell potential \citep{hart1989} underpins both the closed-form characterization and the non-cooperative implementation of the Shapley value.
Its construction relies on balanced contributions holding for \emph{all} pairs of players \citep[see][for a survey of equivalent properties]{casajus2018}---precisely the pair-wise condition \BCplus that we have shown to be incompatible with component efficiency in the presence of externalities (Corollary~\ref{cor:pair-bc-incompatible}).
Similarly, the bidding mechanisms of \citet{perezcastrillo2001} and \citet{slikker2007} use symmetric bidding rounds in which every player in a component can become the proposer; the resulting equilibrium conditions invoke balanced contributions for all connected pairs, not just adjacent ones.
The same applies to bidding mechanisms for partition function form games that extend the P\'{e}rez-Castrillo--Wettstein framework, such as \citet{ju2025}.
Since the BCE rule satisfies balanced contributions only on edges, these standard techniques do not directly extend to our setting.

Whether the BCE rule admits alternative characterizations---such as a marginal-contribution formula, a dividend representation, or a non-cooperative implementation based on tree-structured negotiation protocols---remains an open question for future research.

\clearpage
%% ============================================================
\appendix

%% ============================================================
\section{Notation}\label{app:notation}
%% ============================================================

{\small\linespread{1.0}\selectfont
	\renewcommand{\arraystretch}{1.15}
	\begin{longtable}{@{}lp{9cm}@{}}
		\toprule
		\textbf{Symbol}                                          & \textbf{Meaning}                                                                                                                                    \\
		\midrule
		\endfirsthead
		\multicolumn{2}{r}{\textit{(continued from previous page)}}                                                                                                                                                    \\[2pt]
		\toprule
		\textbf{Symbol}                                          & \textbf{Meaning}                                                                                                                                    \\
		\midrule
		\endhead
		\midrule
		\multicolumn{2}{r}{\textit{(continued on next page)}}                                                                                                                                                          \\
		\endfoot
		\bottomrule
		\endlastfoot
		\multicolumn{2}{@{}l}{\textit{Players and networks}}                                                                                                                                                           \\[2pt]
		$N = \{1,\ldots,n\}$                                     & Finite player set                                                                                                                                   \\
		$g$                                                      & Network (set of unordered pairs / links)                                                                                                            \\
		$g^N$                                                    & Complete network on~$N$                                                                                                                             \\
		$\mathcal{G}$                                            & Set of all networks on~$N$                                                                                                                          \\
		$g|_S$                                                   & Subnetwork of~$g$ induced by $S \subseteq N$: $\{\{k,l\}\in g : k,l\in S\}$                                                                         \\
		$S/g$                                                    & Partition of $S$ into connected components of~$g|_S$                                                                                                \\
		$N/g$                                                    & Partition of $N$ into connected components of~$g$                                                                                                   \\
		$C$                                                      & A connected component ($C \in N/g$)                                                                                                                 \\
		$g_{-D}$                                                 & Network with all links incident to $D \subseteq N$ removed                                                                                          \\
		$g_{-j}$                                                 & Shorthand for $g_{-\{j\}}$ (isolate player~$j$)                                                                                                     \\
		$g \setminus \{i,j\}$                                    & Network with single link $\{i,j\}$ removed                                                                                                          \\
		$g_{-i,-j}$                                              & Shorthand for $g_{-\{i,j\}}$; equals $g_{-j,-i}$                                                                                                    \\[6pt]
		\multicolumn{2}{@{}l}{\textit{Worth functions (all $\to\mathbb{R}$) and allocation rules ($\to\mathbb{R}^N$)}}                                                                                                 \\[2pt]
		$\bar{v}: 2^N \to \mathbb{R}$                            & TU game: $\bar{v}(S)$ = worth of coalition $S\subseteq N$; for the Myerson value only $\bar{v}(C)$ at connected $C$ enters                          \\
		$(C,\mathcal{P})$                                        & Embedded coalition: partition~$\mathcal{P}$ of~$N$ with $C\in\mathcal{P}$                                                                           \\
		$v: \{(C,\mathcal{P})\} \to \mathbb{R}$                  & PFF game: $v(C,\mathcal{P})$ = worth of embedded coalition                                                                                          \\
		$w: \{(C,g):C\in N/g\} \to \mathbb{R}$                   & Network worth function: $w(C,g)$ = worth of component~$C$ in network~$g$                                                                            \\
		$\mathcal{W}$                                            & Space of all network worth functions                                                                                                                \\
		$\varphi: \mathcal{W}\times\mathcal{G} \to \mathbb{R}^N$ & Allocation rule; $\varphi_i(w,g)$ = payoff to player~$i$                                                                                            \\[4pt]
		\multicolumn{2}{@{}l}{\quad\textit{TU games}}                                                                                                                                                                  \\[2pt]
		$\bar{v}^g(S)$                                           & Graph-restricted TU game: $\sum_{D\in S/g} \bar{v}(D)$                                                                                              \\
		$\Sh$                                                    & Shapley value \citep{shapley1953}                                                                                                                   \\
		$\My$                                                    & Myerson value: $\Sh(\bar{v}^g)$ \citep{myerson1977b,myerson1980}                                                                                    \\[4pt]
		\multicolumn{2}{@{}l}{\quad\textit{Network worth functions (J-W and BCE); generalise TU games: $w(C,g)$ may depend on $g$}}                                                                                    \\[2pt]
		$w^g(C,h)$                                               & Graph-projected worth function: $\sum_{S\in C/g} w(S,\,g|_{N/h})$                                                                                   \\
		$\jw$                                                    & J-W value: $\jw_i(w,g)=\Sh_i(\bar{v}^{JW})$, $\bar{v}^{JW}(S)=\sum_{C\in S/g}w(C,g)$ \citep{jackson1996}                                            \\
		$\fce$                                                   & FCE rule: $\Phi^{\preccurlyeq}(v_{w^g})$, $v_{w^g}(C,\mathcal{P})=\sum_{S\in C/g}w(S,g|_{\mathcal{P}})$ \citep{navarro2007}                         \\
		$\bce$                                                   & BCE rule, see~\eqref{eq:bc-tree}--\eqref{eq:formula-full}                                                                                           \\[4pt]
		\multicolumn{2}{@{}l}{\quad\textit{PFF games}}                                                                                                                                                                 \\[2pt]
		$g|_{\mathcal{P}}$                                       & Network restricted to components of~$\mathcal{P}$: $\bigcup_{B\in\mathcal{P}}g|_B$                                                                  \\
		$g/\mathcal{P}$                                          & Partition of $N$ into $g$-components within each block of~$\mathcal{P}$: $\bigcup_{B\in\mathcal{P}} B/g$                                            \\
		$\Phi^{\preccurlyeq}$                                    & Myerson's PFFG value formula \citep{myerson1977a}                                                                                                   \\
		$\Phi^{\mathrm{EF}}$                                     & Externality-free value: $\Sh(\bar{v}^{\mathrm{EF}})$,\ $\bar{v}^{\mathrm{EF}}(S)=v\!\left(S,\{S\}\cup\{\{j\}:j\notin S\}\right)$ \citep{phamdo2007} \\
		$v^g(C,\mathcal{P})$                                     & Graph-restricted PFF game: $\sum_{S\in C/g} v(S,\, g/\mathcal{P})$ (Section~\ref{sec:pff})                                                          \\
		$\Psi^{\mathrm{FCE}}(v,g)$                               & FCE allocation rule for PFF games with a network: $\Phi^{\preccurlyeq}(v^g)$ \citep{feldman1996,navarro2007}                                        \\
		$\Psi^{\mathrm{BCE}}(v,g)$                               & BCE allocation rule for PFF games with a network: $\bce(w_{v},g)$ (Section~\ref{sec:pff})                                                           \\[6pt]
		\multicolumn{2}{@{}l}{\textit{Spanning tree and BCE construction}}                                                                                                                                             \\[2pt]
		$r = \min(C)$                                            & Root of the minimal-index BFS tree for component~$C$                                                                                                \\
		$p(j)$                                                   & Parent of non-root player~$j$ in the BFS tree                                                                                                       \\
		$\gamma_j$                                               & Offset for player~$j$ (equation~\eqref{eq:bc-tree})                                                                                                 \\[6pt]
		\multicolumn{2}{@{}l}{\textit{Cycle-sum identity}}                                                                                                                                                             \\[2pt]
		$i_0,\ldots,i_{\ell-1}$                                  & Players in a fundamental cycle of length~$\ell$                                                                                                     \\
		$e = \{i_a,i_{a+1}\}$                                    & Cycle edge                                                                                                                                          \\
		$D$                                                      & Subset of cycle players not in~$e$                                                                                                                  \\
		$R^{\varphi}_{w}(i,j;\, g')$                             & \BC residual for link~$\{i,j\}$ at network~$g'$ (equation~\eqref{eq:bc-residual})                                                                   \\
		\\
	\end{longtable}
}

\clearpage
\bibliographystyle{plainnat}

\end{document}